\documentclass[12pt]{article} 

\pdfoutput=1

\usepackage[right=2.5cm]{geometry}
\usepackage{tabularx}
\usepackage{amsthm}
\usepackage[table]{xcolor}
\usepackage{mathtools}
\usepackage{newtxtext,newtxmath} 
\usepackage{relsize} 
\usepackage{bm} 
\usepackage{multirow}
\usepackage{booktabs}
\usepackage[absolute]{textpos}

\DeclarePairedDelimiter{\norm}{\lVert}{\rVert}

\DeclareMathOperator*{\argA}{arg}

\makeatletter
\def\thmheadbrackets#1#2#3{%
	\thmname{#1}\thmnumber{\@ifnotempty{#1}{ }\@upn{#2}}%
	\thmnote{ {\the\thm@notefont[#3]}}}
\makeatother

\newtheoremstyle{defbrakets}
{}
{}
{\normalfont}
{}
{\bfseries}
{.}
{ }
{\thmheadbrackets{#1}{#2}{#3}}

\newtheorem{thm}{Theorem}

\newtheorem{lem}[thm]{Lemma}
\theoremstyle{definition}
\theoremstyle{defbrakets}
\newtheorem{defn}{Definition}
\newtheorem{rem}{Remark}
\newtheorem{plm}{Problem}

\newcommand{\z}[1]{\textcolor{black}{#1}}

\begin{document} 
	
	\title{Complete coordination of robotic fiber positioners for massive spectroscopic surveys\footnote{This work was financially supported by the Swiss National Science Foundation (SNF) grant number 20FL21\_185771 and the SLOAN ARC/EPFL agreement number SSP523.}}
	
	\author{Matin Macktoobian$^{a}$, Denis Gillet$^{a}$, and Jean-Paul Kneib$^{b}$}
	\date{%
		$^a$School of Engineering, Swiss Federal Institute of Technology in Lausanne (EPFL), 1015 Lausanne, Switzerland\\%
		$^b$School of Basic Sciences, Swiss Federal Institute of Technology in Lausanne (EPFL), 1015 Lausanne, Switzerland\\[2ex]%
		Corresponding author: Matin Macktoobian (matin.macktoobian@epfl.ch) 
	}

\begin{textblock}{14}(1,1)
	\noindent\textbf{\color{red}Published in ``Journal of Astronomical Telescopes, Instruments, and Systems'' DOI: 10.1117/1.JATIS.5.4.045002}
\end{textblock}

\maketitle

\begin{abstract}
Robotic fiber positioners play a vital role in the generation of massive spectroscopic surveys. The more complete a positioners set is coordinated, the more information its corresponding spectrograph receives during an observation. The complete coordination problem of positioners sets is studied in this paper. We first define the local and the global completeness problems and determine their relationship. We then propose a new artificial potential field according to which the convergences of a positioner and its neighboring positioners are cooperatively taken into account. We also discover the required condition for a complete coordination. We finally explain how the modifications of some of the parameters of a positioners set may resolve its incompleteness coordination scenarios. We verify our accomplishments using simulations.
\end{abstract}

\textbf{keywords}: spectroscopic surveys, robotic fiber positioners, astrobots, coordination, telescopes, focal plane, optimization

\section{Introduction}
Automation\footnote{Throughout this paper, scalars and matrices are represented by regular and bold symbols, respectively.} and robotics have been at the service of space applications for a long time to accomplish different tasks including cargo transportation \cite{kistler2006system}, instrumentation \cite{rochus2007new}, exploration \cite{maimone2007two}, etc. A well-known example is the mobile servicing system \cite{stieber1999overview} mounted at the international space station. This manipulator system executes critical on-orbit assembly tasks and contributes to the external maintenance of the station. Space rovers \cite{iagnemma2004online,macktoobian2013time,macktoobian2016time} represent another category of robotic artifacts which have been extensively used in planetary exploration missions. Observational astronomy benefits from the space robotics as well. Because of the high costs and the safety-critical nature of space applications, autonomy has been taken into account from the earlier days of the space age for orbital observation purposes. For example, the free flyer engineering gave rise to the realization of advanced space telescopes \cite{papadopoulos1994dynamics}. On the other hand, the primary control systems corresponding to ground telescopes were relatively simple, and the level of the required autonomy was not as complicated as that of space telescopes. However, the requirements of recent observation projects need the development of the ground telescopes with higher degrees of autonomy and functional efficiency. In particular, the current trend of astronomy seeks the generation of the whole map of the observable universe using ground telescopes. Supplied with such a cosmological blueprint, geometrical characterization of the universe facilitates the better understanding of the expansion of the universe and the distribution of dark energy all over it. For this purpose, the generation of sky surveys based on the spectroscopic approaches \cite{mazets1982cosmic} has been taken into account.

Traditional telescopes have successfully supported observational operations. However, new requirements of the astronomy, as explained above, require the ground telescopes to contribute to the collection of spectroscopic surveys in a more versatile manner. ``SLOAN Digital Sky Survey'' (SDSS) \cite{york2000sloan} is a set of projects which aim to develop new telescopes observing the evolution of the universe based on cosmic multi-object spectrographs\cite{zhao2012lamost,de20124most}. The current generation of this project is SDSS-V\cite{macktoobian2019navigation} which is under active development. In particular, a spectrograph\cite{takada2014extragalactic} potentially encompasses thousands of optical fibers located at the focal plane of a host telescope in a specific geometrical configuration. Each optical fiber is assigned to a specific target in the sky to be observed by collecting a particular range of the electromagnetic spectral information, particularly visible light, corresponding to that target. The desired range may be visible light, infra red, etc. Since each observation assigns a different target to an optical fiber, a planar RR robotic positioner system is attached to each optical fiber to rotationally move it and to reach its target position located at its configuration space. To maximize the number of the observed objects during each observation, one would like to maximize the number of the mounted fiber positioners at the focal plane of their host telescope. Such a dense hexagonal formation of positioners gives rise to a non-trivial coordination problem for their trajectory planning and collision avoidance. 

Reconfiguration \cite{qasim2018model} refers to the systematic switchings of various configurations of a system each of which exhibits a specific set of functionalities. In particular, coordination problem is a specific subclass of the reconfiguration problem which is vastly studied in different areas including power systems \cite{abdelaziz2010distribution}, hybrid systems \cite{wang2016highly}, discrete-event systems \cite{macktoob2017auto}, consensus of multi-agent systems \cite{olfati2007consensus}, etc. \z{Supervisory control theory\cite{ramadge1987supervisory} was used to seek complete coordinations of robotic fiber positioners\cite{macktoobian2019supervisory}. The major hurdle to use this approach is the curse of dimensionality when the size of a robotic fiber positioners system grows. Then, the required processing is not practically feasible to find a complete solution.} To be specific, the coordination of robotic fiber positioners is challenging because any solution to this problem has to fulfill some critical requirements in both spatial and temporal perspectives. In particular, positioners are often arranged in hexagonal formations, so each positioner neighbors 6 other peers. Furthermore thanks to the applied miniaturization to the manufacturing process of small positioners \cite{horler2018high}, the workspace of a positioner overlap those of its neighboring positioners. Thus, the collision avoidance rises as a major issue to be solved. On the other hand, the on-time coordination of the positioners set is desired after finishing a specific observation to point to the objects of the next observation. Since each observation is extremely time-dependent, the coordination of the system shall be executed in a limited amount of time between two successive observations. Thus, the solution to the coordination problem of robotic fiber positioners has to be both reliable against collisions and efficient in view of performance.

The solutions to the trajectory planning and the collision avoidance problems directly depend on the number and the mechanical specifications of the used positioners in a particular subproject of the SDSS project. For example in the case of the ``The Dark Spectroscopic Instrument'' (DESI) \cite{1aghamousa2016desi,2aghamousa2016desi} project, an artificial potential field (APF) approach is proposed to solve the collision-free trajectory planning of positioners. This method uses a decentralized navigation function based on the notion of artificial potential fields. In particular, the arms of the positioners used in this project are long enough to enter the workspace of any neighboring positioners. However since the contentions are not considerable, all the positioners can converge to their target positions. In other words, the positioners of the DESI project compose a complete system. In contrast, the completeness is not realized in the case of "The Multi Object Optical and Near-infrared Spectrograph" (MOONS) \cite{cirasuolo2014moons} project. In this case, the length of the second arm of each positioner is two times longer than those of the positioners of the DESI project. To solve the trajectory planning problem associated with the MOONS project, the planning algorithm was modified \cite{tao2018priority} to take two subjects into account. First, not every colliding situations is managed by the navigation function. So, a priority-based decision-making layer was added to the decentralized navigation function to handle deadlocks and oscillations which could not be handled by the navigation function. Based on this approach, the positioners which are assigned to more important objects are prioritized in the coordination of the system. Thus, some positioners may not reach their target at all. The algorithm cannot generally coordinate the system such that all positioners reach their target positions. In other words, the coordination problem is not complete with respect to the solutions of this algorithm. Complete coordination leads to the collection of the full information which is planned to be collected during an observation. However, no analysis has been yet applied to explore the conditions based on which a solution to a coordination problem of positioners is complete. This gap opens an avenue for the potential modification of the current coordination algorithm to realize the complete coordination of positioners. 

In this paper, we formally analyze and solve the complete coordination problem associated with robotic optical fiber positioners. We obtain a completeness condition whose fulfillment guarantees the complete coordination. The remainder of the report is organized and follows. Sec. \ref{sec:mecSpec} briefly reviews the mechanical specifications of a typical positioner. Sec. \ref{sec:locGlob} establishes the global completeness problem whose solution shall guarantee the convergence of all positioners of a telescope. We then define the local completeness problem corresponding to the convergence of a positioner and all of its neighboring positioners. In particular, we take a distributed scheme into account to show that given a set of positioners, if all local completeness problems corresponding to neighboring region of the system are complete, then the overall system is globally complete. Sec. \ref{sec:CAPF} proposes a new class of artificial potential fields, i.e., cooperative artificial potential fields (CAPFs). The advantage of a CAPF compared to an APF is that the attractive term of the CAPF considers not only the convergence of its own positioner agent but also the convergence of its neighboring positioners. Thanks to the proved solvability of the global completeness problem based on the completeness of its local completeness problems, Sec. \ref{subsec:compCond} obtains the required condition for the solvability of the local completeness problem. Sec. \ref{subsec:compSeek} establishes a strategy for completeness seeking when a system of positioners is incomplete with respect to a particular set of parameter specifications of the system. In these situations, we indeed propose to modify the paramaters corresponding to the specification of the system's CAPFs and/or the definition of the desired observation to resolve the encountered incompleteness. Sec. \ref{sec:diss} compares CAPF to APF in view of the properties of the navigation process such as computational complexity and convergence time. We evaluate our accomplishments by simulations in Sec. \ref{sec:sim}. Our concluding remarks are finally drawn in Sec. \ref{sec:conc}.
\section{Mechanical Characterization}
\label{sec:mecSpec}
This section follows a top-down approach to briefly introduce cosmic spectroscopy and robotic fiber positioners. In particular, we first study the process of observation based on spectrographs. We particularly describe the role of robotic fiber positioners in the quoted process. Then, we present details about the mechanical structure of a typical robotic fiber positioner and its kinematic formulation.
\begin{figure}
	\centering
	\includegraphics[scale=1.3]{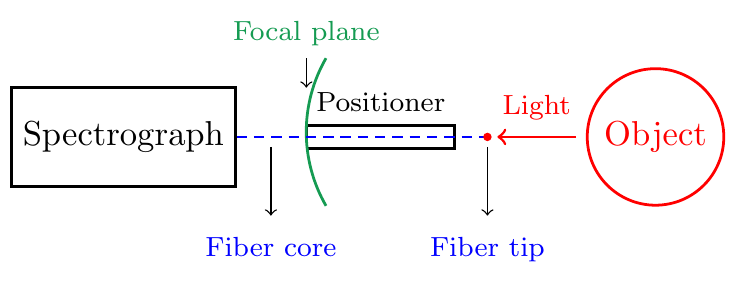}
	\caption{The schematic of a telescope equipped with a robotic optical fiber positioner}
	\label{fig:1}
\end{figure}

Massive spectroscopic surveys are generated by collecting spectral information coming from massive sets of objects by telescopes. Then, the information is processed by a spectrograph to construct a unified map of those objects. Fig. \ref{fig:1} illustrates a typical observation task of a single robotic fiber positioner. In particular, a string of optical fiber is passed through a robotic positioner. The tip of the optical fiber can be moved by the motors of the robotic positioner to point to a specific object in the sky. The robotic positioner is mounted at a curved plate called focal plane inside the telescope. At the back of the focal plane, the optical fiber is connected to a spectrograph. Fianlly, the spectrograph processes the received signal from the tip of the fiber which yields the generation of the desired map.   

Each positioner is a planar RR manipulator whose end-effector shall reach the point at which its fiber has to observe an object based on a particular observation.
The forward kinematics corresponding to the workspace of each positioner is described as below.
\begin{equation}
\label{eq:kin}
\bm{q^{i}} = \bm{q^{i}_{b}} + \begin{bmatrix}
\cos{(\theta^{i})} & \cos{(\theta^{i} + \phi^{i})}\\
\sin{(\theta^{i})} & \sin{(\theta^{i} + \phi^{i})}
\end{bmatrix}\bm{l}
\end{equation}
Here the $i$\textsuperscript{th} positioner is located at $\bm{q^{i}} = \begin{bmatrix}x^{i} & y^{i}\end{bmatrix}^\intercal$ with respect to a universal frame attached to the
focal plane of the host telescope. $\bm{q^{i}_{b}}= \begin{bmatrix}x^{i}_{b} & y^{i}_{b}\end{bmatrix}^\intercal$ is also the base coordination of the positioner. The lengths of rotational links are represented by $\bm{l} = \begin{bmatrix}l_{1} &l_{2}\end{bmatrix}^\intercal$. The angular positions of the $i$\textsuperscript{th} positioner are denoted by $\theta^{i}$ and $\phi^{i}$. The quoted parameters are depicted in Fig. \ref{fig:2}(a). Fig. \ref{fig:2}(b) illustrates the positioners placement in the focal plane of a typical telescope.	
\begin{figure}
	\centering
	\includegraphics[scale=1]{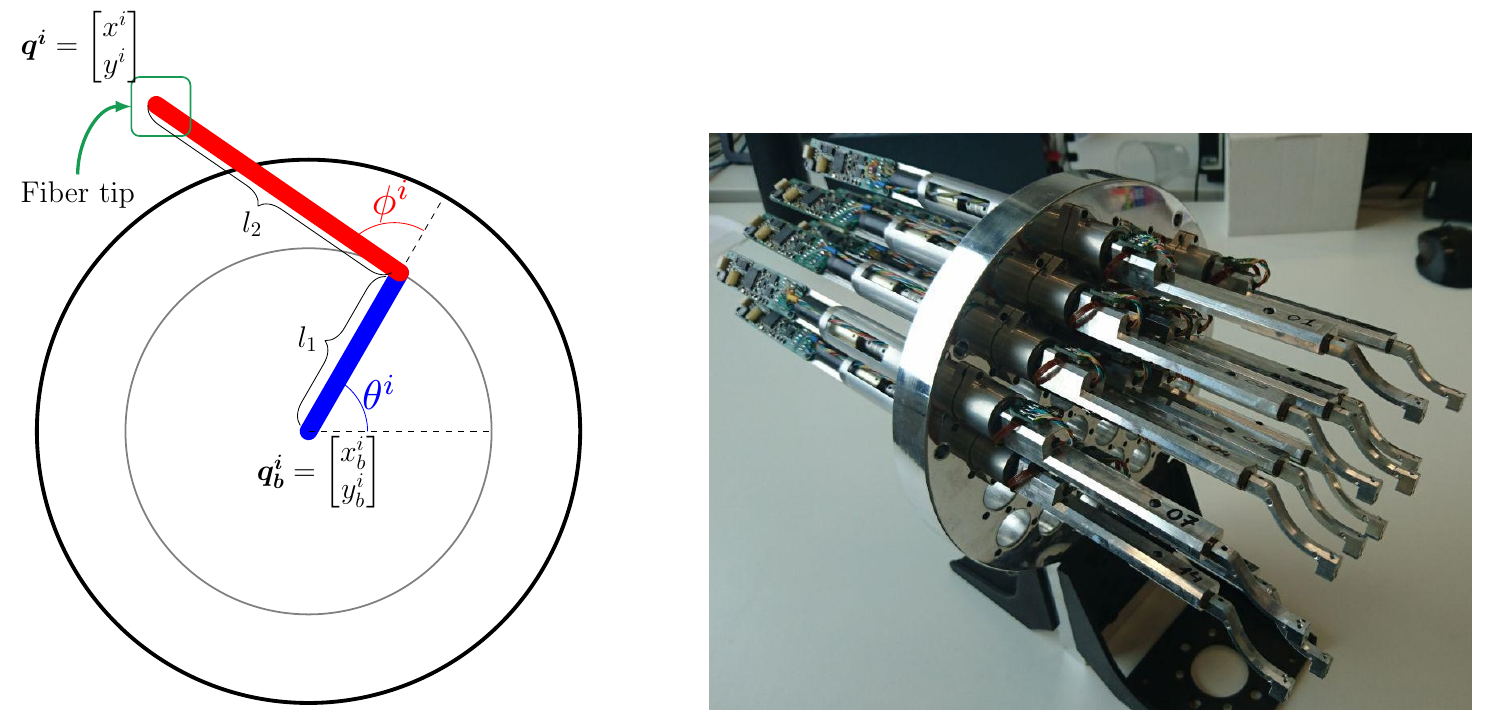}
	\caption{A typical robotic optical fiber positioner system (reprinted from \cite{macktoobian2019navigation} with permission). (a) The $\theta-\phi$ RR design of a typical positioner. (b) The hexagonal arrangement of a set of positioners in a focal plane of a telescope}
	\label{fig:2}
\end{figure}

The focal plane area of a telescope is composed of a set $\mathcal{P}$ of fiber positioners as depicted in Fig. \ref{fig:2}(b). Each observation includes a set of target objects each of which should be observed by a fiber positioner. A system of positioners is called completely coordinated (or complete) if all assigned positioners can be aligned to point to their targets as depicted in Fig. \ref{fig:3}(a); otherwise the system is partially coordinated (or incomplete) as rendered by Fig. \ref{fig:3}(b). Completeness does matter since the information throughput of an observation is maximum when its observing positioner system is completely coordinated. In the next section, we formally define the completeness problem to be later solved in upcoming sections. 
\begin{figure}[t]
	\centering
	\includegraphics[scale=1]{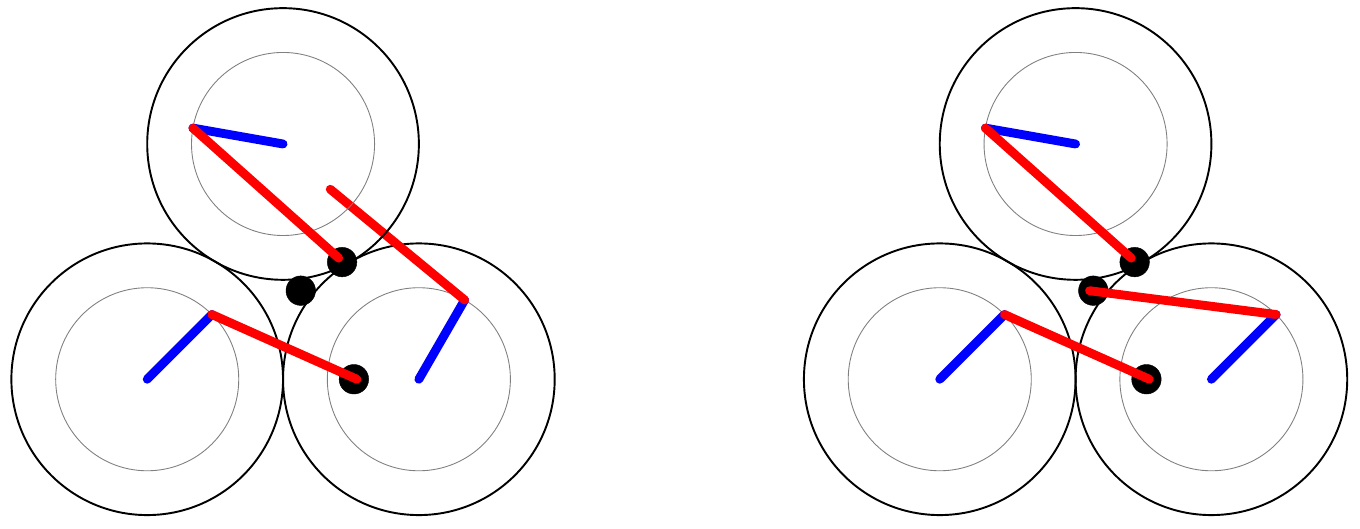}
	\caption{The visual illustration of the notion of completeness (The black circles represent the target spots corresponding to the postioners.). (a) A typical incomplete coordination. (b) A complete coordination example}
	\label{fig:3}
\end{figure}
\section{From Local to Global Completeness}	
\label{sec:locGlob}
In this section, we define the global and
the local completeness problems. We then show that the solvability of a global completeness problem is equivalent to the solvability of all the local	completeness problems associated with it.

We first define the notions of ``position'', ``target position'', and ''equilibrium position`` corresponding to a positioner as follows.	
\begin{defn}[Position]
	Let $\mathcal{P}$ be the set of all the positioners associated with a telescope. Given a positioner $\pi^{i} \in \mathcal{P}$, $\bm{q^{i}} := \begin{bmatrix}x^{i} &y^{i}\end{bmatrix}^\intercal$ denotes the (planar) \textit{position} of $\pi^{i}$. 
\end{defn}
\begin{defn}[Target Position]
	Let $\mathcal{P}$ be the set of all the positioners associated with a telescope. Given a positioner $\pi^{i} \in \mathcal{P}$, $\bm{q^{i}_{\mathcal{T}}}:=\begin{bmatrix}x^{i}_{\mathcal{T}} & y^{i}_{\mathcal{T}}\end{bmatrix}^\intercal$ represents the \textit{target position} at which $\pi^{i}$ is planned to reach according to a specific observation.
\end{defn}
\begin{defn}[Equilibrium Position]
	Let $\mathcal{P}$ be the set of all the positioners associated with a telescope. Given a positioner $\pi^{i} \in \mathcal{P}$, $\bm{q^{i}_{\star}}:= \begin{bmatrix}x^{i}_{\star} & y^{i}_{\star}\end{bmatrix}^\intercal$ represents the \textit{equilibrium position} at which $\pi^{i}$ resides at the end of the coordination process.
\end{defn}
\begin{defn}[Index Set]
	Let $\mathcal{N}:=\{\pi^{1}, \cdots, \pi^{n}\}$ be the set of $n$ positioners. Then, $\mathcal{I}_{\mathcal{N}}$ is the \textit{index set} of $\mathcal{N}$ denoting the set of all the indices of the elements of $\mathcal{N}$ as follows
	\begin{equation*}
	\mathcal{I}_{\mathcal{N}} := \{\argA\limits_{k} \pi^{k} | \forall \pi^{k} \in \mathcal{N}\},
	\end{equation*} 
	where $\argA(\cdot)$ operator returns the index of its arguments.
\end{defn}  
Now we define the ``global completeness problem'' as follows.
\begin{plm}[Global Completeness]
	Subject to a set of positioners $\mathcal{P}$ and its corresponding index set $\mathcal{I}_{\mathcal{P}}$, determine whether or not the following relation holds.
	\begin{equation*}
	(\forall k \in \mathcal{I}_{\mathcal{P}}) \bm{q^{k}_{\star}} = \bm{q^{k}_{\mathcal{T}}}
	\end{equation*}
\end{plm}
Because of the dense hexagonal arrangements of positioners in a focal plane, the direct solution to the problem above would be difficult. Instead, we define a local version of the completeness problem, and we show that how the solutions to a set of local completeness problems end up with the solution to the global completeness problem corresponding to them. For this purpose, we first define the notion of "neighborhood" with respect to a specific positioner.
\begin{defn}[Neighborhood]
	Let $\mathcal{P}$ be the set of all the positioners associated with a telescope. Let $\pi^{i} \in \mathcal{P}$ be a positioner. Given $\mathcal{V}^{i}$ denoting the neighboring positioners of $\pi^{i}$, $\mathcal{N}^{i} \subseteq \mathcal{P}$ is the \textit{neighborhood} with respect to $\pi^{i}$ defined as follows
	\begin{equation*}
	\mathcal{N}^{i} := \mathcal{V}^{i} \dot{\bigcup} \{\pi^{i}\}.
	\end{equation*}
\end{defn}
The following definition establishes the ``local completeness problem''.
\begin{plm}[Local Completeness]
	Let $\mathcal{P}$ be the set of all the positioners associated with a telescope. Subject to the neighborhood $\mathcal{N}^{i} \subseteq \mathcal{P}$ with respect to a positioner $\pi^{i} \in \mathcal{P}$, determine whether or not the following holds.
	\begin{equation*}
	(\forall k \in \mathcal{I}_{\mathcal{N}^{i}}) \bm{q^{k}_{\star}} = \bm{q^{k}_{\mathcal{T}}}
	\end{equation*}
\end{plm}  
Using the definition above, we establish the notion of ``completeness relation''
\begin{defn}[Completeness Relation]
	Let $\mathcal{P}$ be the set of all the positioners associated with a telescope. Let also $\mathcal{N}^{i} \subseteq \mathcal{P}$ be a neighborhood with respect to the positioner $\pi^{i} \in \mathcal{P}$. Then, if $\mathcal{N}^{i}$ is locally complete, then the following relation holds.
	\begin{equation*}
	\mathcal{C}(\mathcal{N}^{i})
	\end{equation*}
	As well, given the set of all neighborhoods $\mathcal{N}$ corresponding to positioners set $\mathcal{P}$, if $\mathcal{N}$ is globally complete then $\mathcal{C}(\mathcal{N})$ holds.
\end{defn}
We prove the following property of the completeness relation which is subsequently used to show the relationship between the notions of local and global completeness.
\begin{lem}
	\label{lem:closed}
	Completeness relation is closed under countable union operator.
\end{lem}
\begin{proof}
	Let $\mathcal{N}^{i}$ and $\mathcal{N}^{j}$ be two locally complete neighborhoods with respect to positioners $\pi^{i}$ and $\pi^{j}$, respectively, i.e., $\mathcal{C}(\mathcal{N}^{i})$ and $\mathcal{C}(\mathcal{N}^{j})$ hold. Then, the following two cases shall be mutual exclusively considered.
	\begin{itemize}
		\item $\mathcal{N}^{i}$ and $\mathcal{N}^{j}$ are not adjacent to each other, say,
		\begin{equation*}
		(\forall \pi \in \mathcal{N}^{i}) \pi \not\in \mathcal{N}^{j}.
		\end{equation*}
		Accordingly, there is no interaction between the quoted neighborhoods. So, the every positioner also reaches its target position after the unification of the neighborhoods. Therefore, the resulting union in complete.
		\item $\mathcal{N}^{i}$ and $\mathcal{N}^{j}$ are adjacent to each other, say,
		\begin{equation*}
		(\exists \pi \in \mathcal{N}^{i}) \pi \in \mathcal{N}^{j}.
		\end{equation*}
		In a hexagonal arrangement of positioners, the minimum and the maximum numbers of the shared positioners\footnote{The maximum number of the shared positioners varies with respect to the lengths of the positioners' arms. The reader finds a thorough analysis of the cited relationship in \cite{horler2018robotic}} between two adjacent neighborhoods are 1 and 3, respectively. Considering the minimum case, let $\pi$ be the shared positioner, so it is the the exclusive positioner which can potentially disturb the overall completeness of $\mathcal{N}^{i}$ and $\mathcal{N}^{j}$. However according to the assumption of the completeness of both neighborhoods, $\pi$ shall reach its target in view of both neighborhoods. Thus, the unification of the neighborhoods is complete. The similar argument is valid to justify the completeness of the unified system of complete neighborhoods where the number of shared events is 2 or 3, as well.
	\end{itemize}
\end{proof}
Finally, the following theorem uses Lemma \ref{lem:closed} to state how the local and the global completeness problems are related to each other. 
\begin{thm}
	\label{thm:equiv}
	Let $\mathcal{N}$ be the set of all neighborhoods to which the positioners of a telescope are assigned. So, if all neighborhoods are locally complete, then the overall system of the positioners is complete, i.e.,
	\begin{equation*}
	[(\forall \mathcal{N}^{i} \in \mathcal{N}) \mathcal{C}(\mathcal{N}^{i})] \Rightarrow \mathcal{C}(\mathcal{N}).
	\end{equation*}
\end{thm}
\begin{proof}
	By induction, we show that the proof is a consequence of Lemma\ref{lem:closed}. In particular, let $k$ be the number of the neighborhoods. Then, we have
	\begin{itemize}
		\item base case: $k = 1$, say, the positioners set includes only one (complete) neighborhood. So, the overall system is obviously complete.
		\item induction step: suppose the system with $k=n$ is complete, i.e., $\mathcal{C}(\bigcup\limits_{i=1}^{n}\mathcal{N}^{i})$ holds. We show that the system with $k = n+1$ has to be complete. In particular given complete neighborhood $\mathcal{N}^{n+1}$, since the completeness relation is closed under countable union operator (see, Lemma \ref{lem:closed}), hence we have
		\begin{equation*}
		\begin{aligned}
		\mathcal{C}(\bigcup\limits_{i=1}^{n}\mathcal{N}^{i}) \cup \mathcal{C}(\mathcal{N}^{n+1}) &= \mathcal{C}(\bigcup\limits_{i=1}^{n}\mathcal{N}_{i} \cup \mathcal{N}^{n+1})  \\
		&= \mathcal{C}(\bigcup\limits_{i=1}^{n+1}\mathcal{N}^{i}),
		\end{aligned}
		\end{equation*}
		which concludes the global completeness of the positioners set. 
	\end{itemize}
\end{proof}
we later (see, Sec. \ref{sec:analysis}) analyze the completeness condition for local systems. Thanks to the result of Theorem \ref{thm:equiv}, if the conditions corresponding to the completeness of all of the neighborhoods are hold, then the global system is also complete. In the next section, we rewrite the definitions of the local and the global completeness problems in the language of artificial potential fields (APFs). Then, we revise the formulation of the decentralized navigation function, used in priority-based coordination method \cite{tao2018priority}. So, the equilibrium of the new APF could represent the complete result of a coordination process. We then uncover the condition for the existence of a solution to the local completeness problem. 
\section{Cooperative Artificial Potential Fields}
\label{sec:CAPF}
In this section, we first explain our motivation to define a new type of APFs, called ``cooperative artificial potential field'' (CAPF). In particular, we elaborate on the effect of an APF on the completeness of the coordination process. In particular, Sec. \ref{subsec:motiv} clarifies our angle of attack to tackle the completeness problem. Then, we formally introduce our proposed CAPF in Sec. \ref{subsec:formul}. We also reformulate the local and the global completeness problems using the notion of CAPF.
\subsection{Motivation}
\label{subsec:motiv}
The priority-based algorithm \cite{tao2018priority} uses a two-layer competitive architecture to solve the coordination problem, as depicted in Fig. \ref{fig:4}.
\begin{figure}
	\centering
	\includegraphics[scale=1.3]{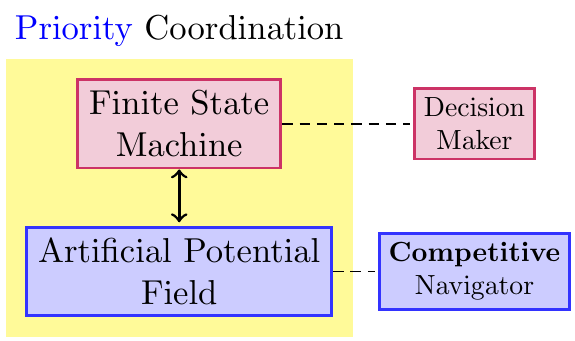}
	\caption{The competitive control architecture of the priority-based coordination}
	\label{fig:4}
\end{figure}

Let $\lambda_{1}$ and $\lambda_{2}$ be positive constant weighting factors. Let also $D$ be the radius of the collision avoidance envelope in which the repulsive term is activated. $d$ also represents the radius of the safety region around each positioner. Then, we define
\begin{subequations}
	\begin{align*}
	\bm{\lambda_{1}} := \lambda_{1} \otimes \mathbb{I}_{2},\\
	\bm{\z{\lambda_{2}}} := \z{\lambda_{2}} \otimes \mathbb{I}_{2}.
	\end{align*}
\end{subequations}
Thus, the definition of the reference APF used in \cite{tao2018priority} is represented as follows
\begin{equation}
\label{eq:typical_apf}
\psi(\bm{q^{i}}) := \underbrace{\vphantom{\bm{\lambda_{2}}\displaystyle\sum\limits_{\mathclap{j \in \mathcal{I}_{\mathcal{N}^{i}}\setminus \{i\}}} \min [0,\frac{\norm{\bm{q^{i}} - \bm{q^{j}}}^{2} - D^{2}}{\norm{\bm{q^{i}} - \bm{q^{j}}}^{2} - d^{2}}}\bm{\lambda_{1}}\norm{\bm{q^{i}} - \bm{q^{i}_{\mathcal{T}}}}^{2}}_{\text{attractive term}} + \underbrace{\bm{\lambda_{2}}\displaystyle\sum\limits_{\mathclap{j \in \mathcal{I}_{\mathcal{N}^{i}}\setminus \{i\}}} \min~ \left[\bm{0},\frac{\norm{\bm{q^{i}} - \bm{q^{j}}}^{2} - D^{2}}{\norm{\bm{q^{i}} - \bm{q^{j}}}^{2} - d^{2}}\right]}_{\text{repulsive term}}.
\end{equation} 		
One notes that the attractive term above exclusively takes the convergence of the positioner $\pi^{i}$ into account. So, the APFs corresponding to different positioners in fact compete with each other because each artificial potential field only cares about the convergence of its own affiliated positioner. Since a positioner does not care about the convergence of its neighbors, this competitive manner of navigation potentially gives rise to the incomplete coordination of the overall system of positioners.

Instead, we propose a cooperative scheme based on which each positioner not only seeks its own convergence, but also cares about the convergence of its neighboring counterparts. Thus, the competitive architecture can be modified based on this cooperative perspective as depicted in Fig. \ref{fig:5}.	
\begin{figure}
	\centering
	\includegraphics[scale=1.3]{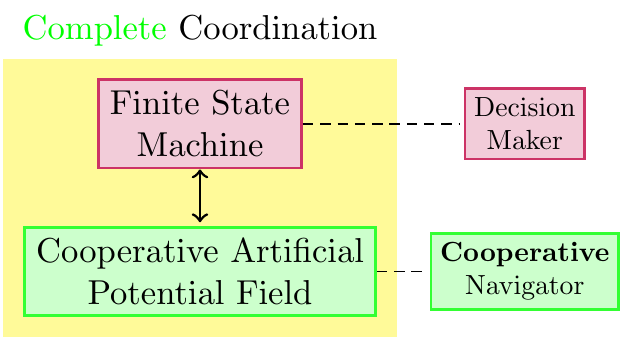}
	\caption{The cooperative control architecture of the complete coordination}
	\label{fig:5}
\end{figure}	
\subsection{Formulation}
\label{subsec:formul}
We embed a particular attractive term in the definition of the reference APF (see, Eq. (\ref{eq:typical_apf})) to realize the cooperation between neighboring positioners to reach collective convergence to their target spots. Let $\lambda_{3}$ be a positive weighting factor corresponding to the cooperative attractive term. Let also $\bm{q^{i}_{\mathcal{T}}}$ (resp. $\bm{q^{j}_{\mathcal{T}}}$) be the target position of $\bm{q^{i}}$ (resp. $\bm{q^{j}}$). Then considering
\begin{equation*}
\bm{\lambda_{3}} := \lambda_{3} \otimes \mathbb{I}_{2},
\end{equation*}
we define a CAPF into which a cooperative attractive term is integrated as follows.
\begin{equation}
\label{eq:CAPF}
\xi(\bm{q^{i}}) := \underbrace{\vphantom{\bm{\lambda_{2}}\displaystyle\sum\limits_{j \in \mathcal{I}_{\mathcal{N}^{i}}\setminus \{i\}} \min [0,\frac{\norm{\bm{q^{i}} - \bm{q^{j}}}^{2} - D^{2}}{\norm{\bm{q^{i}} - \bm{q^{j}}}^{2} - d^{2}}}\bm{\lambda_{1}}\norm{\bm{q^{i}} - \bm{q^{i}_{\mathcal{T}}}}^{2}}_{\text{attractive term}} + \underbrace{\bm{\lambda_{3}}\displaystyle\sum\limits_{\mathclap{j \in \mathcal{I}_{\mathcal{N}^{i}}\setminus \{i\}}}\norm{\bm{q^{j}} - \bm{q^{j}_{\mathcal{T}}}}^{2}}_{\text{cooperative attractive term}} +  \underbrace{\bm{\lambda_{2}}\displaystyle\sum\limits_{\mathclap{j \in \mathcal{I}_{\mathcal{N}^{i}}\setminus \{i\}}} \min~ \left[\bm{0},\frac{\norm{\bm{q^{i}} - \bm{q^{j}}}^{2} - D^{2}}{\norm{\bm{q^{i}} - \bm{q^{j}}}^{2} - d^{2}}\right]}_{\text{repulsive term}}
\end{equation}   	
The cooperative attractive term inserts extra dynamics to the reference APF to involve all positioners of a neighborhood in the convergence process. A rough guideline to set the value of $\lambda_{3}$ is $\lambda_{3} < \lambda_{1}$ for two reasons. First, each CAPF instance should mainly focus on the convergence of its corresponding positioner. So, one selects a larger weighting factor for the main positioner to insure that the main portion of the attractive force of its corresponding CAPF comes from that positioner. Second, $\lambda_{3}$ in fact injects the velocity profile of the neighboring positioners to that of the main positioner. Any large values corresponding to those velocity profiles may give rise to abrupt motions imposed to the main positioner. Such unwanted and uncontrolled motions may not only damage the main positioner's actuators but also leave it vulnerable to potential collisions.  	

We are interested in the conditions based on which a solution to a specific coordination problem is complete. Thus, we formulate the local and the global completeness problems in the language of CAPF. In particular, the equilibrium points for all positioners in a neighborhood shall be their target points. Since the positioners exclusively stop moving at their target points, one needs to obtain the equilibrium points corresponding to the derivative of CAPF as follows.	
\begin{equation}
\hspace*{-2cm}
\label{eq:nabla-xi}
\nabla \xi(\bm{q^{i}}) = 
\begin{cases}
2\bm{\lambda_{1}}(\bm{q^{i}} - \bm{q^{i}_{\mathcal{T}}}) + 2\bm{\lambda_{3}}\displaystyle\sum\limits_{\mathclap{j \in \mathcal{I}_{\mathcal{N}^{i}}\setminus \{i\}}}(\bm{q^{j}} - \bm{q^{j}_{\mathcal{T}}}) & (\forall j \in \mathcal{I}_{\mathcal{N}^{i}}\setminus \{i\}) \norm{\bm{q^{i}} - \bm{q^{j}}} \ge D\\
2\bm{\lambda_{1}}(\bm{q^{i}} - \bm{q^{i}_{\mathcal{T}}}) + 2\bm{\lambda_{3}}\displaystyle\sum\limits_{\mathclap{j \in \mathcal{I}_{\mathcal{N}^{i}}\setminus \{i\}}}(\bm{q^{j}} - \bm{q^{j}_{\mathcal{T}}}) + 2\bm{\lambda_{2}}\displaystyle\sum\limits_{\mathclap{j \in \mathcal{I}_{\mathcal{N}^{i}}\setminus \{i\}}}~~~\frac{(D^{2}-d^{2})(\bm{q^{i}} - \bm{q^{j}})}{\bigl(\norm{\bm{q^{i}} - \bm{q^{j}}}^{2} - d^{2}\bigr)^{2}} & (\exists j \in \mathcal{I}_{\mathcal{N}^{i}}\setminus \{i\}) \norm{\bm{q^{i}} - \bm{q^{j}}} < D
\end{cases}
\end{equation}
According to the forward kinematic model of a typical positioner, i.e., (\ref{eq:kin}), and the CAPF defined in (\ref{eq:CAPF}), the control law below is proposed to be applied to the joints of the positioner $\pi^{i}$. 
\begin{equation*}
\bm{u^{i}} := -\nabla_{\theta_{i},\phi_{i}}\xi(\bm{q^{i}})
\end{equation*}	
Now we can compose the CAPF-driven formalism of the local and the global completeness problems as follows.	
\begin{plm}[Local Completeness (CAPF Derivation)]
	\label{plm:loc-nav}
	Let $\mathcal{N}^{i}$ be a neighborhood with respect to the positioner $\pi^{i}$ where $2 \le |\mathcal{N}^{i}| \le 7$. Then, the neighborhood is locally complete coordinated by a set of CAPFs if the following differential equations are simultaneously solvable.
	\begin{equation*}
	\label{eq:loc_cond}
	\nabla \xi(\bm{q^{i}}) = 0 ~~~ \text{for}~~~ 1 \le i \le |\mathcal{N}^{i}|
	\end{equation*}
\end{plm}
The global completeness problem is the generalization of the local completeness problem above as below. 	
\begin{plm}[Global Completeness (CAPF Derivation)]
	Let $\mathcal{P}$ be the set of all positioners of a telescope. Then, the overall system is globally complete coordinated by a set of CAPFs if the following differential equations are simultaneously solvable.
	\begin{equation*}
	\label{eq:glob_cond}
	\nabla \xi(\bm{q^{i}}) = 0 ~~~ \text{for}~~~ 1 \le i \le |\mathcal{P}|
	\end{equation*}
\end{plm}	
\begin{rem}
	Each CAPF has one exclusive minimum \z{because it is a smooth Morse function \cite{milnor1963morse}, it is uniformly maximal on boundaries of a free space, and it has a unique minimum at a target point in its free space \cite{makarem2015decentralized}}. Then, if Eq. (\ref{eq:loc_cond}) (resp., Eq. (\ref{eq:glob_cond})) is solvable, then its solution is essentially $\bm{q_{\mathcal{T}}} := [\bm{q^{1}_{\mathcal{T}}} ~\dots~ \bm{q^{|\mathcal{N}^{i}|}_{\mathcal{T}}}]^\intercal$(resp., $\bm{q_{\mathcal{T}}} := [\bm{q^{1}_{\mathcal{T}}} ~\dots~ \bm{q^{|\mathcal{P}|}_{\mathcal{T}}}]^\intercal$).
\end{rem}	
In the next section, we find the conditions for guaranteed solvability of the local and the global completeness problems.
\section{Completeness Analysis}
\label{sec:analysis}
The preceding section revealed that the solutions to all local completeness problems are the keys to determine whether the global completeness problem corresponding to those problems is solvable. Here Sec. \ref{subsec:compCond} focuses on the required condition for the completeness of a neighborhood. Then, Sec. \ref{subsec:compSeek} discusses a procedure based on which completeness is sought regarding a system of positioners which is not complete according to a particular set of parameters. 
\subsection{Completeness Condition}
\label{subsec:compCond}	
We take a typical isolated neighborhood with the maximum number of positioners, say, $\{\pi^{i}|0 \le i \le 6\}$. We also consider the maximum contention between the positioners of the neighborhood. In particular, we assume that two neighboring positioners, e.g., $\pi^{1}$ and $\pi^{2}$, are at the collision zone of the central positioner, i.e., $\pi^{0}$. The remaining four positioners are assumed to be residing at each other's collision zones in a pair-wise manner, say, $\pi^{3}$ and $\pi^{4}$, and $\pi^{5}$ and $\pi^{6}$. Fig. \ref{fig:6} represents the configuration of the neighborhood, in which the regions with the same color correspond to those positioners which are suspected to collide and to block each other's movements. This scenario is the most collision-susceptible case to reach the full completeness for the explained neighborhood.  	
\begin{figure}
	\centering
	\includegraphics[scale=0.4]{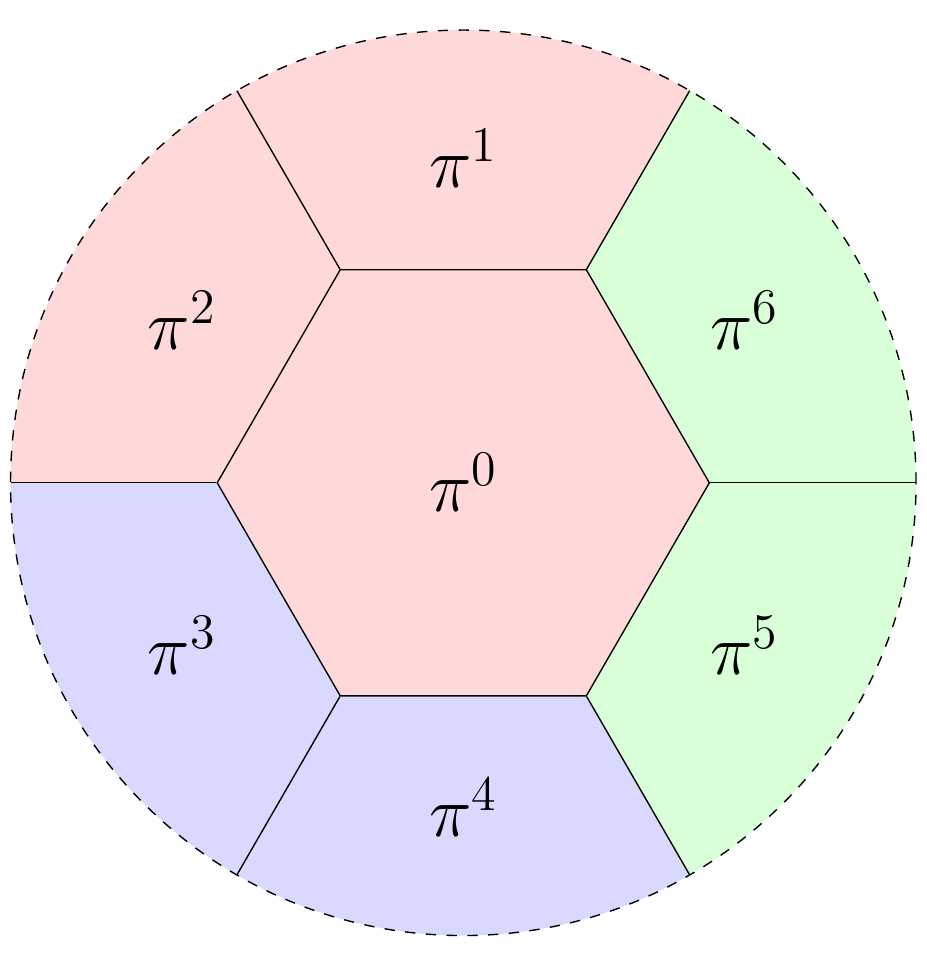}
	\caption{The arrangement of positioners in a typical neighborhood subject to the maximum contention (The regions with the same color correspond to those positioners which are suspected to collide and to block each other's movements.)}
	\label{fig:6}
\end{figure}

According to Problem \ref{plm:loc-nav}, we need to find the solutions which simultaneously fulfill the following set of equations.	
\begin{equation}
\label{eq:reference-eq}
\nabla \xi(\bm{q^{i}}) = \bm{0} ~~~ \text{for} ~~~ 0 \le i \le 6
\end{equation}	
Using Eq. (\ref{eq:nabla-xi}), we expand Eq. (\ref{eq:reference-eq}) as follows.	
	\begin{align*}
	\nabla \xi(\bm{q^{0}}) &= 2\bm{\lambda_{1}}(\bm{q^{0}} - \bm{q^{0}}_{\mathcal{T}}) + 2\bm{\lambda_{3}}\displaystyle\sum\limits_{\mathclap{j \in \{1,2\}}}(\bm{q^{j}} - \bm{q^{j}_{\mathcal{T}}}) + 2\bm{\lambda_{2}}(D^{2} - d^{2})\displaystyle\sum\limits_{\mathclap{j \in \{1,2\}}}~\frac{\bm{q^{0}} - \bm{q^{j}}}{{\bigl(\norm{\bm{q^{0}} - \bm{q^{j}}}^{2} - d^{2}\bigr)^{2}}} = \bm{0}	\\
	\nabla \xi(\bm{q^{1}}) &= 2\bm{\lambda_{1}}(\bm{q^{1}} - \bm{q^{1}_{\mathcal{T}}}) + 2\bm{\lambda_{3}}\displaystyle\sum\limits_{\mathclap{j \in \{0,2\}}}(\bm{q^{j}} - \bm{q^{j}_{\mathcal{T}}}) + 2\bm{\lambda_{2}}(D^{2} - d^{2})\displaystyle\sum\limits_{\mathclap{j \in \{0,2\}}}~\frac{\bm{q^{1}} - \bm{q^{j}}}{{\bigl(\norm{\bm{q^{1}} - \bm{q^{j}}}^{2} - d^{2}\bigr)^{2}}} = \bm{0} \\
	\nabla \xi(\bm{q^{2}}) &= 2\bm{\lambda_{1}}(\bm{q^{2}} - \bm{q^{2}_{\mathcal{T}}}) + 2\bm{\lambda_{3}}\displaystyle\sum\limits_{\mathclap{j \in \{0,1\}}}(\bm{q^{j}} - \bm{q^{j}_{\mathcal{T}}}) + 2\bm{\lambda_{2}}(D^{2} - d^{2})\displaystyle\sum\limits_{\mathclap{j \in \{0,1\}}}~\frac{\bm{q^{\z{2}}} - \bm{q^{j}}}{{\bigl(\norm{\bm{q^{2}} - \bm{q^{j}}}^{2} - d^{2}\bigr)^{2}}}  = \bm{0} \\
	\nabla \xi(\bm{q^{3}}) &= 2\bm{\lambda_{1}}(\bm{q^{3}} - \bm{q^{3}_{\mathcal{T}}}) + 2\bm{\lambda_{3}}(\bm{q^{4}} - \bm{q^{4}_{\mathcal{T}}}) + 2\bm{\lambda_{2}}(D^{2} - d^{2})\dfrac{\bm{q^{3}} - \bm{q^{4}}}{{\bigl(\norm{\bm{q^{3}} - \bm{q^{4}}}^{2} - d^{2}\bigr)^{2}}}  = \bm{0} \\
	\nabla \xi(\bm{q^{4}}) &= 2\bm{\lambda_{1}}(\bm{q^{4}} - \bm{q^{4}_{\mathcal{T}}}) + 2\bm{\lambda_{3}}(\bm{q^{3}} - \bm{q^{3}_{\mathcal{T}}}) + 2\bm{\lambda_{2}}(D^{2} - d^{2})\dfrac{\bm{q^{4}} - \bm{q^{3}}}{{\bigl(\norm{\bm{q^{4}} - \bm{q^{3}}}^{2} - d^{2}\bigr)^{2}}}  = \bm{0} \\
	\nabla \xi(\bm{q^{5}}) &= 2\bm{\lambda_{1}}(\bm{q^{5}} - \bm{q^{5}_{\mathcal{T}}}) + 2\bm{\lambda_{3}}(\bm{q^{6}} - \bm{q^{6}_{\mathcal{T}}}) + 2\bm{\lambda_{2}}(D^{2} - d^{2})\dfrac{\bm{q^{5}} - \bm{q^{6}}}{{\bigl(\norm{\bm{q^{5}} - \bm{q^{6}}}^{2} - d^{2}\bigr)^{2}}}  = \bm{0} \\
	\nabla \xi(\bm{q^{6}}) &= 2\bm{\lambda_{1}}(\bm{q^{6}} - \bm{q^{6}_{\mathcal{T}}}) + 2\bm{\lambda_{3}}(\bm{q^{5}} - \bm{q^{5}_{\mathcal{T}}}) + 2\bm{\lambda_{2}}(D^{2} - d^{2})\dfrac{\bm{q^{6}} - \bm{q^{5}}}{{\bigl(\norm{\bm{q^{6}} - \bm{q^{5}}}^{2} - d^{2}\bigr)^{2}}}  = \bm{0}
	\end{align*}
To compact the set of equations above, we define the following auxiliary function	
\begin{equation*}
\mathcal{Q}(\bm{q^{\alpha}},\bm{q^{\beta}}) := \dfrac{\bm{q^{\alpha}} - \bm{q^{\beta}}}{{\bigl(\norm{\bm{q^{\alpha}} - \bm{q^{\beta}}}^{2} - d^{2}\bigr)^{2}}},
\end{equation*} 	
and the constant parameter below	
\begin{equation*}
\omega := D^{2} - d^{2},
\end{equation*}	
which yield	
	\begin{align*}
	\nabla \xi(\bm{q^{0}}) &= 2\begin{bmatrix}
		\bm{\lambda_{1}} & \bm{\lambda_{3}} & \bm{\lambda_{3}}
	\end{bmatrix}\begin{bmatrix}
		\bm{q^{0}} & \bm{q^{1}} & \bm{q^{2}}
	\end{bmatrix}^\intercal + 2\omega\bm{\lambda_{2}}\bigl(\mathcal{Q}(\bm{q^{0}},\bm{q^{1}})+\mathcal{Q}(\bm{q^{0}},\bm{q^{2}})\bigr) - 2\bigl(\bm{\lambda_{1}}\bm{q^{0}_{\mathcal{T}}} + \bm{\lambda_{3}}(\bm{q^{1}_{\mathcal{T}}} + \bm{q^{2}_{\mathcal{T}}})\bigr) = \bm{0},\\
	\nabla \xi(\bm{q^{1}}) &= 2\begin{bmatrix}
		\bm{\lambda_{3}} & \bm{\lambda_{1}} & \bm{\lambda_{3}}
	\end{bmatrix}\begin{bmatrix}
		\bm{q^{0}} & \bm{q^{1}} & \bm{q^{2}}
	\end{bmatrix}^\intercal + 2\omega\bm{\lambda_{2}}\bigl(\mathcal{Q}(\bm{q^{1}},\bm{q^{0}})+\mathcal{Q}(\bm{q^{1}},\bm{q^{2}})\bigr) - 2\bigl(\bm{\lambda_{1}}\bm{q^{1}_{\mathcal{T}}} + \bm{\lambda_{3}}(\bm{q^{0}_{\mathcal{T}}} + \bm{q^{2}_{\mathcal{T}}})\bigr) = \bm{0},\\
	\nabla \xi(\bm{q^{2}}) &= 2\begin{bmatrix}
		\bm{\lambda_{3}} & \bm{\lambda_{3}} & \bm{\lambda_{1}}
	\end{bmatrix}\begin{bmatrix}
		\bm{q^{0}} & \bm{q^{1}} & \bm{q^{2}}
	\end{bmatrix}^\intercal + 2\omega\bm{\lambda_{2}}\bigl(\mathcal{Q}(\bm{q^{2}},\bm{q^{0}})+\mathcal{Q}(\bm{q^{2}},\bm{q^{1}})\bigr) - 2\bigl(\bm{\lambda_{1}}\bm{q^{2}_{\mathcal{T}}} + \bm{\lambda_{3}}(\bm{q^{0}_{\mathcal{T}}} + \bm{q^{1}_{\mathcal{T}}})\bigr) = \bm{0},\\
	\nabla \xi(\bm{q^{3}}) &= 2\begin{bmatrix}
		\bm{\lambda_{1}} & \bm{\lambda_{3}}
	\end{bmatrix}\begin{bmatrix}
		\bm{q^{3}} & \bm{q^{4}}
	\end{bmatrix}^\intercal + 2\omega\bm{\lambda_{2}}\mathcal{Q}(\bm{q^{3}},\bm{q^{4}}) - 2\bigl(\bm{\lambda_{1}}\bm{q^{3}_{\mathcal{T}}} + \bm{\lambda_{3}}\bm{q^{4}_{\mathcal{T}}}\bigr) = \bm{0},\\
	\nabla \xi(\bm{q^{4}}) &= 2\begin{bmatrix}
		\bm{\lambda_{3}} & \bm{\lambda_{1}}
	\end{bmatrix}\begin{bmatrix}
		\bm{q^{3}} & \bm{q^{4}}
	\end{bmatrix}^\intercal + 2\omega\bm{\lambda_{2}}\mathcal{Q}(\bm{q^{4}},\bm{q^{3}}) - 2\bigl(\bm{\lambda_{1}}\bm{q^{4}_{\mathcal{T}}} + \bm{\lambda_{3}}\bm{q^{3}_{\mathcal{T}}}\bigr) = \bm{0},\\
	\nabla \xi(\bm{q^{5}}) &= 2\begin{bmatrix}
		\bm{\lambda_{1}} & \bm{\lambda_{3}}
	\end{bmatrix}\begin{bmatrix}
		\bm{q^{5}} & \bm{q^{6}}
	\end{bmatrix}^\intercal + 2\omega\bm{\lambda_{2}}\mathcal{Q}(\bm{q^{5}},\bm{q^{6}}) - 2\bigl(\bm{\lambda_{1}}\bm{q^{5}_{\mathcal{T}}} + \bm{\lambda_{3}}\bm{q^{6}_{\mathcal{T}}}\bigr) = \bm{0},\\
	\nabla \xi(\bm{q^{6}}) &= 2\begin{bmatrix}
		\bm{\lambda_{3}} & \bm{\lambda_{1}}
	\end{bmatrix}\begin{bmatrix}
		\bm{q^{5}} & \bm{q^{6}}
	\end{bmatrix}^\intercal + 2\omega\bm{\lambda_{2}}\mathcal{Q}(\bm{q^{6}},\bm{q^{5}}) - 2\bigl(\bm{\lambda_{1}}\bm{q^{6}_{\mathcal{T}}} + \bm{\lambda_{3}}\bm{q^{5}_{\mathcal{T}}}\bigr) = \bm{0}.
	\end{align*}
The equations set above can be written as follows	
\begin{equation*}
\hspace*{-3cm}
\label{eq:vec1} 
\underbrace{\begin{bmatrix}
		\nabla \xi(\bm{q^{0}}) \\ \nabla \xi(\bm{q^{1}}) \\ \nabla \xi(\bm{q^{2}}) \\ \nabla \xi(\bm{q^{3}}) \\	\nabla \xi(\bm{q^{4}}) \\ \nabla \xi(\bm{q^{5}}) \\ \nabla \xi(\bm{q^{6}}) 
\end{bmatrix}}_{\textstyle\mathlarger{\mathlarger{\mathlarger{\nabla \xi(\bm{q})}}}} = 
\underbrace{\begin{bmatrix}
		2\bm{\lambda_{1}}&2\bm{\lambda_{3}}&2\bm{\lambda_{3}} &\bm{0}&\bm{0}&\bm{0}&\bm{0}\\
		2\bm{\lambda_{3}}&2\bm{\lambda_{1}}&2\bm{\lambda_{3}} &\bm{0}&\bm{0}&\bm{0}&\bm{0}\\
		2\bm{\lambda_{3}}&2\bm{\lambda_{3}}&2\bm{\lambda_{1}} &\bm{0}&\bm{0}&\bm{0}&\bm{0}\\
		\bm{0}&\bm{0}&\bm{0}&2\bm{\lambda_{1}}&2\bm{\lambda_{3}}&\bm{0}&\bm{0}\\
		\bm{0}&\bm{0}&\bm{0}&2\bm{\lambda_{3}}&2\bm{\lambda_{1}}&\bm{0}&\bm{0}\\
		\bm{0}&\bm{0}&\bm{0}&\bm{0}&\bm{0}&2\bm{\lambda_{1}}&2\bm{\lambda_{3}}\\
		\bm{0}&\bm{0}&\bm{0}&\bm{0}&\bm{0}&2\bm{\lambda_{3}}&2\bm{\lambda_{1}}	
\end{bmatrix}}_{\textstyle\mathlarger{\mathlarger{\mathlarger{\bm{\Lambda}}}}}
\underbrace{\begin{bmatrix}
		\bm{q^{0}}\\\bm{q^{1}}\\\bm{q^{2}}\\\bm{q^{3}}\\\bm{q^{4}}\\\bm{q^{5}}\\\bm{q^{6}}
\end{bmatrix}}_{\textstyle\mathlarger{\mathlarger{\mathlarger{\bm{q}}}}}
+2\omega\bm{\lambda_{2}}
\underbrace{\begin{bmatrix}
		\mathcal{Q}(\bm{q^{0}},\bm{q^{1}}) + \mathcal{Q}(\bm{q^{0}},\bm{q^{2}})\\
		\mathcal{Q}(\bm{q^{1}},\bm{q^{0}}) + \mathcal{Q}(\bm{q^{1}},\bm{q^{2}})\\
		\mathcal{Q}(\bm{q^{2}},\bm{q^{0}}) + \mathcal{Q}(\bm{q^{2}},\bm{q^{1}})\\
		\mathcal{Q}(\bm{q^{3}},\bm{q^{4}})\\
		\mathcal{Q}(\bm{q^{4}},\bm{q^{3}})\\
		\mathcal{Q}(\bm{q^{5}},\bm{q^{6}})\\
		\mathcal{Q}(\bm{q^{6}},\bm{q^{5}})
\end{bmatrix}}_{\textstyle\mathlarger{\mathlarger{\mathlarger{\bm{\Omega}}}}}		
+
\underbrace{\begin{bmatrix}
		2\bigl(\bm{\lambda_{1}}\bm{q^{0}_{\mathcal{T}}} + \bm{\lambda_{3}}(\bm{q^{1}_{\mathcal{T}}} + \bm{q^{2}_{\mathcal{T}}})\bigr)\\
		2\bigl(\bm{\lambda_{1}}\bm{q^{1}_{\mathcal{T}}} + \bm{\lambda_{3}}(\bm{q^{0}_{\mathcal{T}}} + \bm{q^{2}_{\mathcal{T}}})\bigr)\\
		2\bigl(\bm{\lambda_{1}}\bm{q^{2}_{\mathcal{T}}} + \bm{\lambda_{3}}(\bm{q^{0}_{\mathcal{T}}} + \bm{q^{1}_{\mathcal{T}}})\bigr)\\
		2\bigl(\bm{\lambda_{1}}\bm{q^{3}_{\mathcal{T}}} + \bm{\lambda_{3}}\bm{q^{4}_{\mathcal{T}}}\bigr)\\
		2\bigl(\bm{\lambda_{1}}\bm{q^{4}_{\mathcal{T}}} + \bm{\lambda_{3}}\bm{q^{3}_{\mathcal{T}}}\bigr)\\
		2\bigl(\bm{\lambda_{1}}\bm{q^{5}_{\mathcal{T}}} + \bm{\lambda_{3}}\bm{q^{6}_{\mathcal{T}}}\bigr)\\
		2\bigl(\bm{\lambda_{1}}\bm{q^{6}_{\mathcal{T}}} + \bm{\lambda_{3}}\bm{q^{5}_{\mathcal{T}}}\bigr)
\end{bmatrix}}_{\textstyle\mathlarger{\mathlarger{\mathlarger{\bm{\Theta'}}}}} 
= \bm{0},
\end{equation*}
whose compact form reads
\begin{equation*}
\label{eq:vec1-compact}
\nabla \xi(\bm{q}) = \bm{\Lambda}\bm{q} + 2\omega\bm{\lambda_{2}}\bm{\Omega} + \bm{\Theta'} = \bm{0}.
\end{equation*}
The entries of $\bm{\Omega}$ above include function $\mathcal{Q}(\cdot,\cdot)$ which is nonlinear. We note that both positioners monotonically head to their target points. So as an approximation, we linearize this function at the point whose coordinates are the average of the target positions' coordinates associated with the arguments of the function. Put differently, we linearize $\mathcal{Q}(\bm{q^{\alpha}},\bm{q^{\beta}})$ at $\begin{bmatrix}
	\frac{\bm{q^{\alpha}_{\mathcal{T}}}+\bm{q^{\beta}_{\mathcal{T}}}}{2} & \frac{\bm{q^{\alpha}_{\mathcal{T}}}+\bm{q^{\beta}_{\mathcal{T}}}}{2}\end{bmatrix}^\intercal$ which is the closest point to both positioners. Thus, the Newton method gives the following approximation for $\mathcal{Q}(\cdot,\cdot)$.	
\begin{equation*}
\hspace*{-2cm}
\label{eq:approx}
\mathcal{Q}(\bm{q^{\alpha}},\bm{q^{\beta}}) \approx \mathcal{Q}(\bm{q^{\alpha}_{\mathcal{T}}},\bm{q^{\beta}_{\mathcal{T}}}) + \frac{\partial\mathcal{Q}(\bm{q^{\alpha}},\bm{q^{\beta}})}{\partial \bm{q^{\alpha}}}\bigg\rvert_{\bigl(\frac{\bm{q^{\alpha}_{\mathcal{T}}}+\bm{q^{\beta}_{\mathcal{T}}}}{2},\frac{\bm{q^{\alpha}_{\mathcal{T}}}+\bm{q^{\beta}_{\mathcal{T}}}}{2}\bigr)}\bigl(\bm{q^{\alpha}} - \frac{\bm{q^{\alpha}_{\mathcal{T}}}+\bm{q^{\beta}_{\mathcal{T}}}}{2}\bigr) + \frac{\partial\mathcal{Q}(\bm{q^{\alpha}},\bm{q^{\beta}})}{\partial \bm{q^{\beta}}}\bigg\rvert_{\bigl(\frac{\bm{q^{\alpha}_{\mathcal{T}}}+\bm{q^{\beta}_{\mathcal{T}}}}{2},\frac{\bm{q^{\alpha}_{\mathcal{T}}}+\bm{q^{\beta}_{\mathcal{T}}}}{2}\bigr)}\bigl(\bm{q^{\beta}} - \frac{\bm{q^{\alpha}_{\mathcal{T}}}+\bm{q^{\beta}_{\mathcal{T}}}}{2}\bigr)
\end{equation*}  	
Taking the auxiliary constant parameters below into account	
\begin{subequations}
	\begin{align*}
	&\bm{\overline{\Delta}_{\alpha,\beta}} = \bm{\overline{\Delta}_{\beta,\alpha}} : = \frac{\bm{q^{\alpha}_{\mathcal{T}}}+\bm{q^{\beta}_{\mathcal{T}}}}{2},\\
	&\bm{\Delta^{\alpha}_{\alpha,\beta}} = \bm{\Delta^{\alpha}_{\beta,\alpha}} := \frac{\partial\mathcal{Q}(\bm{q^{\alpha}},\bm{q^{\beta}})}{\partial \bm{q^{\alpha}}}\bigg\rvert_{\bigl(\bm{\overline{\Delta}_{\alpha,\beta}},\bm{\overline{\Delta}_{\alpha,\beta}}\bigr)},\\
	&\bm{\Delta^{\beta}_{\alpha,\beta}} = \bm{\Delta^{\beta}_{\beta,\alpha}} := \frac{\partial\mathcal{Q}(\bm{q^{\alpha}},\bm{q^{\beta}})}{\partial \bm{q^{\beta}}}\bigg\rvert_{\bigl(\bm{\overline{\Delta}_{\alpha,\beta}},\bm{\overline{\Delta}_{\alpha,\beta}}\bigr)},
	\end{align*}
\end{subequations}
(\ref{eq:approx}) is simplified as below	
\begin{equation*}
\label{eq:approx-simplified}
\mathcal{Q}(\bm{q^{\alpha}},\bm{q^{\beta}}) \approx \mathcal{Q}(\bm{q^{\alpha}_{\mathcal{T}}},\bm{q^{\beta}_{\mathcal{T}}}) - \bm{\overline{\Delta}_{\alpha,\beta}}\bigl(\bm{\Delta^{\alpha}_{\alpha,\beta}} + \bm{\Delta^{\beta}_{\alpha,\beta}}\bigr) + \bm{q^{\alpha}}\bm{\Delta^{\alpha}_{\alpha,\beta}} + \bm{q^{\beta}}\bm{\Delta^{\beta}_{\alpha,\beta}}.
\end{equation*}	
Therefore, the linearized version of $\bm{\Omega}$, i.e., $\bm{\Omega^{\star}}$, is obtained as the following:	
\begin{equation*}
\bm{\Omega} \approx \bm{\Omega^{\star}} = 
\underbrace{\begin{bmatrix}
		\bm{\Delta^{0}_{0,1}} + \bm{\Delta^{0}_{0,2}} &\bm{\Delta^{1}_{0,1}}&\bm{\Delta^{2}_{0,2}}&\bm{0}&\bm{0}&\bm{0}&\bm{0}\\
		\bm{\Delta^{0}_{1,0}}&\bm{\Delta^{1}_{1,0}}+\bm{\Delta^{1}_{1,2}}&\bm{\Delta^{2}_{1,2}}&\bm{0}&\bm{0}&\bm{0}&\bm{0}\\
		\bm{\Delta^{0}_{2,0}}&\bm{\Delta^{1}_{2,1}}&\bm{\Delta^{2}_{2,0}}+\bm{\Delta^{2}_{2,1}}&\bm{0}&\bm{0}&\bm{0}&\bm{0}\\
		\bm{0}&\bm{0}&\bm{0}&\bm{\Delta^{3}_{3,4}}&\bm{\Delta^{4}_{3,4}}&\bm{0}&\bm{0}\\
		\bm{0}&\bm{0}&\bm{0}&\bm{\Delta^{3}_{4,3}}&\bm{\Delta^{4}_{4,3}}&\bm{0}&\bm{0}\\
		\bm{0}&\bm{0}&\bm{0}&\bm{0}&\bm{0}&\bm{\Delta^{5}_{5,6}}&\bm{\Delta^{6}_{5,6}}\\
		\bm{0}&\bm{0}&\bm{0}&\bm{0}&\bm{0}&\bm{\Delta^{5}_{6,5}}&\bm{\Delta^{6}_{6,5}}
\end{bmatrix}}_{\textstyle\mathlarger{\mathlarger{\mathlarger{\bm{\Delta}}}}}
\underbrace{\begin{bmatrix}
		\bm{q^{0}}\\\bm{q^{1}}\\\bm{q^{2}}\\\bm{q^{3}}\\\bm{q^{4}}\\\bm{q^{5}}\\\bm{q^{6}}
\end{bmatrix}}_{\textstyle\mathlarger{\mathlarger{\mathlarger{\bm{q}}}}} + \bm{\Theta''},
\end{equation*}	
where	
\begin{equation*}
\bm{\Theta''} = 
\begin{bmatrix}
	\mathcal{Q}(\bm{q^{0}_{\mathcal{T}}},\bm{q^{1}_{\mathcal{T}}})    +     \mathcal{Q}(\bm{q^{0}_{\mathcal{T}}},\bm{q^{2}_{\mathcal{T}}})    -
	\bm{\overline{\Delta}_{0,1}} (\bm{\Delta^{0}_{0,1}} + \bm{\Delta^{1}_{0,1})}    -
	\bm{\overline{\Delta}_{0,2}} (\bm{\Delta^{0}_{0,2}} + \bm{\Delta^{2}_{0,2}})\\
	\mathcal{Q}(\bm{q^{1}_{\mathcal{T}}},\bm{q^{0}_{\mathcal{T}}})    +     \mathcal{Q}(\bm{q^{1}_{\mathcal{T}}},\bm{q^{2}_{\mathcal{T}}})    -
	\bm{\overline{\Delta}_{0,1}} (\bm{\Delta^{1}_{1,0}} + \bm{\Delta^{0}_{1,0})}    -
	\bm{\overline{\Delta}_{1,2}} (\bm{\Delta^{1}_{1,2}} + \bm{\Delta^{2}_{1,2}})\\
	\mathcal{Q}(\bm{q^{2}_{\mathcal{T}}},\bm{q^{0}_{\mathcal{T}}})    +     \mathcal{Q}(\bm{q^{2}_{\mathcal{T}}},\bm{q^{1}_{\mathcal{T}}})    -
	\bm{\overline{\Delta}_{0,2}} (\bm{\Delta^{0}_{2,0}} + \bm{\Delta^{2}_{2,0}})    -
	\bm{\overline{\Delta}_{1,2}} (\bm{\Delta^{1}_{2,1}} + \bm{\Delta^{2}_{2,1}})\\
	\mathcal{Q}(\bm{q^{3}_{\mathcal{T}}},\bm{q^{4}_{\mathcal{T}}})    -
	\bm{\overline{\Delta}_{3,4}} (\bm{\Delta^{3}_{3,4}} + \bm{\Delta^{4}_{3,4}})\\
	\mathcal{Q}(\bm{q^{4}_{\mathcal{T}}},\bm{q^{3}_{\mathcal{T}}})    -
	\bm{\overline{\Delta}_{3,4}} (\bm{\Delta^{3}_{4,3}} + \bm{\Delta^{4}_{4,3}})\\
	\mathcal{Q}(\bm{q^{5}_{\mathcal{T}}}
	,\bm{q^{6}_{\mathcal{T}}})    -
	\bm{\overline{\Delta}_{5,6}} (\bm{\Delta^{5}_{5,6}} + \bm{\Delta^{6}_{5,6}})\\
	\mathcal{Q}(\bm{q^{6}_{\mathcal{T}}},\bm{q^{5}_{\mathcal{T}}})    -
	\bm{\overline{\Delta}_{6,5}} (\bm{\Delta^{5}_{6,5}} + \bm{\Delta^{6}_{6,5}})
\end{bmatrix}.
\end{equation*}	
We replace $\bm{\Omega}$ in Eq. (\ref{eq:vec1-compact}) by its linear approximation $\bm{\Omega^{\star}}$.	
\begin{equation*}
\begin{aligned}
\nabla \xi(\bm{q}) =& \bm{\Lambda}\bm{q} + 2\omega\bm{\lambda_{2}}\bm{\Omega} + \bm{\Theta'} = \bm{0}\\
\approx& \bm{\Lambda}\bm{q} + 2\omega\bm{\lambda_{2}}\bm{\Omega^{\star}} + \bm{\Theta'} = \bm{0}\\
\approx& \bm{\Lambda}\bm{q} + 2\omega\bm{\lambda_{2}}(\bm{\Delta}\bm{q}+\bm{\Theta''}) + \bm{\Theta'} = \bm{0}\\
\approx& \bm{\Lambda}\bm{q} + \underbrace{2\omega\bm{\lambda_{2}}\bm{\Delta}}_{\bm{\Gamma}}\bm{q}+\underbrace{2\omega\bm{\lambda_{2}}\bm{\Theta''} + \bm{\Theta'}}_{\bm{\Theta}} = \bm{0}\\
\end{aligned}
\end{equation*}	
Thus, we end up with 	
\begin{equation}
\label{eq:lin-eq}
(\bm{\Lambda} + \bm{\Gamma})\bm{q} + \bm{\Theta} = \bm{0}.
\end{equation}	
Now, we can analyze the solvability of the local completeness problem based on Eq. (\ref{eq:lin-eq}), called \textit{the completeness equation}. For a system of positioners to be complete, this equation has be solvable, and its solution has to be the target points corresponding to the positioners of the system. In particular, a system is complete if the following equation holds
\begin{equation*}
\bm{q_\mathcal{T}} = -(\bm{\Lambda} + \bm{\Gamma})^{-1}\bm{\Theta};
\end{equation*} 
otherwise, it is incomplete.

The completeness equation asserts that the completeness of a system of positioners depends on the parameters that are set by designers. Thus, modification of those parameters may resolve any potential incompleteness. For this purpose, in the next section we propose a parameter modification process to search for completeness encountering an incomplete system.
\subsection{Completeness Seeking by Parameter Modification}
\label{subsec:compSeek}
As the completeness equation implies, the parameters which shape $\bm{\Lambda}$ and $\bm{\Gamma}$ directly influence on the completeness of system. Strictly speaking, parameter selections may give rise to incompleteness. So, one can take two approaches into account to search for the parameters based on which the system is complete. Considering an incomplete system with respect to a particular parameter specification, we modify entries of $\bm{\Lambda}$ and/or $\bm{\Gamma}$ to search for the other parameter specifications based on which the system is complete. 

Matrix $\bm{\Lambda}$ is structured by the attractive and the cooperative attractive terms of CAPFs. So, if a system is incomplete, one can change the values corresponding the weighting factors of the cited terms. So, the overall summation of the $\bm{\Lambda}$ and $\bm{\Gamma}$ might be invertible. Theoretically, there are infinitely many numbers which can be attributed to the weighting factors. So, there is no upper bound for the number of the possible parameter modifications corresponding to $\bm{\Lambda}$. However, practical requirements constrain the scope of valid selections. For example, large weighting factors strictly increase the velocity of positioners. The resulting high velocities may damage their motors and increase the possibility of collision when the positioners are very close to each other. Thus, a reasonable range for each weighting factor can be determined from which new values are selected to modify the current values.

Matrix $\bm{\Gamma}$ also contributes to the completeness (on incompleteness) of a system based on its parameters. Among all those parameters, the target positions extremely affect on the entries of the matrix. One may note that, the target positions are defined based on each observation. In particular, some algorithms are used to assign an object to each positioner. For example, \cite{morales2011fibre} handles the object-positioner assignments such that the number of the observed objects is maximized. We note that changing the targets assigned to the positioners ends up with a new matrix $\bm{\Gamma}$. So, such a parameter modification may resolve the system incompleteness. In contrast to the $\bm{\Lambda}$ modification, the maximum number of the possible target position modifications is bounded. As already quoted, a specific procedure assigns a target to each positioner according to a particular observation prior to the coordination. In particular, given $n$ objects corresponding to an observation and $m \ge n$ positioners\footnote{\noindent We assume that an observation is planned such that all of its objects could be observed by the positioners set. Thus, the number of the positioners should essentially exceeds that of those objects.}, the number of possible object-positioner assignments is $P(m,n)$. However, every target cannot be observed by every positioner because of the positioners' motion limitations. Another option to modify $\bm{\Gamma}$ would be changing the value of the repulsive weighting factor, i.e., $\lambda_{2}$. However, manipulation of this factor is not recommended because of its critical role in the safety of the system and its performance. In particular, decreasing the factor may jeopardize the full control over movements of postioners when they are close to each other. In contrast, increasing the value of the factor can extremely increase the required time for completion of the coordination process. The explained parameter modification process is schematically illustrated in Fig. \ref{fig:7}.  
\begin{figure}
	\hspace*{-5mm}
	\includegraphics[scale=0.7]{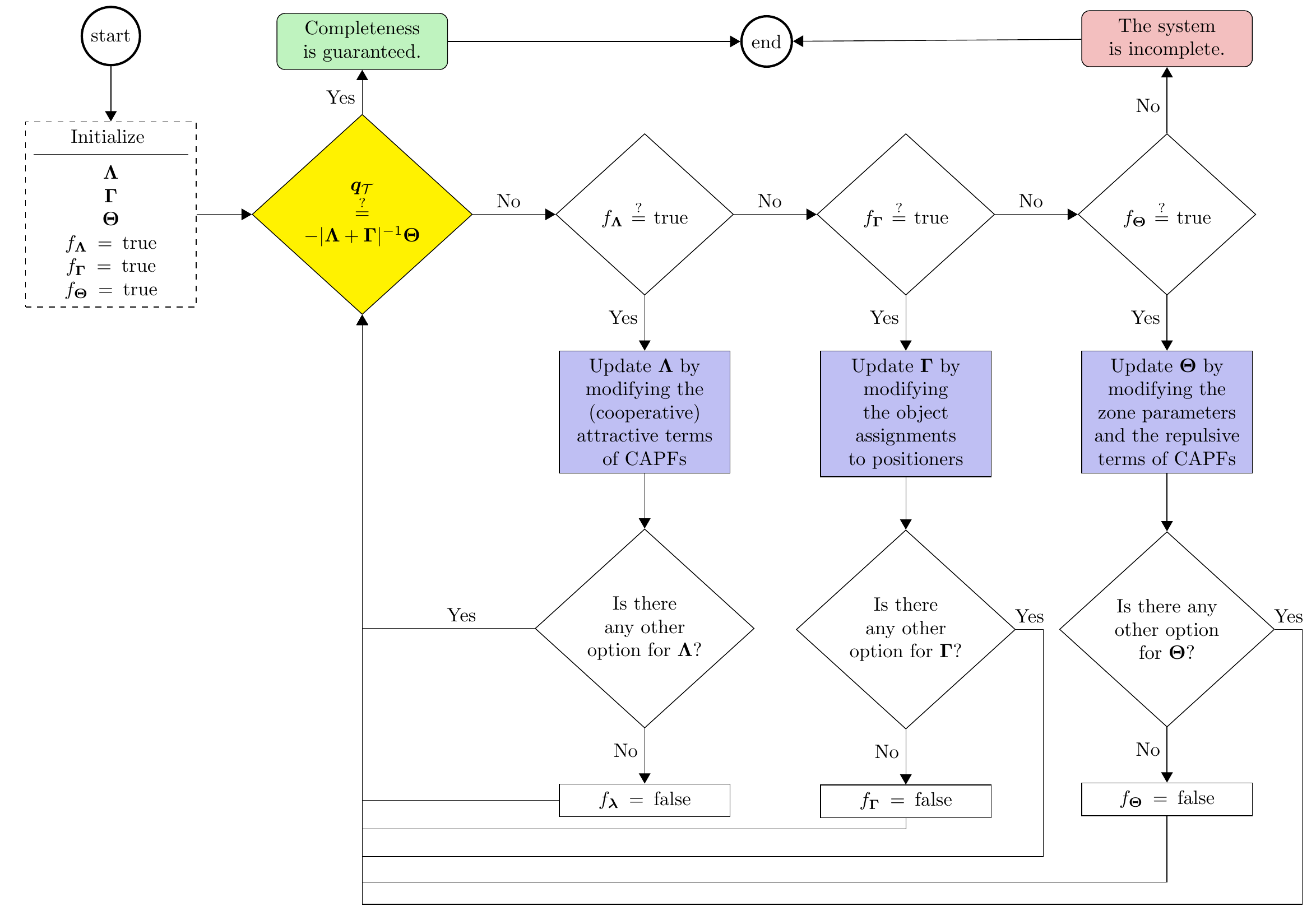}
	\caption{The parameter modification process}
	\label{fig:7}
\end{figure}
\section{Discussion}
\label{sec:diss}	
The computational complexity of the trajectory planning algorithm using the reference APF is $\mathcal{O}(n)$ where $n$ represents the number of the positioners to be coordinated \cite{makarem2014collision}. The substitution of CAPF for APF does not adversely affect the computational complexity of the overall trajectory planning algorithm applied to positioners sets. To be specific, the added cooperative attractive term is a polynomial similar to the attractive term of the algorithm. Thus, the linear-time computational complexity of the algorithm is preserved.

The added cooperative attractive term increases the agility of the movements in the course of coordination. However, this agility has to be compensated and attenuated in practice because abrupt movements of positioners may strengthen the collision possibility when they are close to each other. In other words, the added cooperative attractive term does not necessarily improve the convergence time of the coordination process. Furthermore, in some cases the convergence time might be even longer than that of corresponding to the reference APF. In the case of the reference APF, each positioner stops moving upon reaching its target position. However, in the case of CAPF, a positioner does not necessarily resides at its target spot immediately after reaching it because the cooperative term induces more dynamics to settle the maximum of the neighboring positioners at their target points. Thus, a positioner may temporarily pass its target to open the way for the remainder of its peers to get closer to their targets. This behavior does not give rise to endless oscillations since the high-level decision-making layer in fact handles these kind of scenarios. Thus, using CAPF rises a trade-off between the completeness seeking and potentially longer convergence time. The simulation results of the next section confirms this conclusion.
\section{Simulations}
\label{sec:sim}	
We modify the Python simulator developed in \cite{tao2018priority} according to our contributions. In particular, we substitute the reference APF \cite{tao2018priority} (see, Fig. \ref{fig:4}) with our CAPF. (see, Fig. \ref{fig:5}).

We define two test batches. Each test batch includes \z{six} test scenarios each of which includes a specific number of positioners. Furthermore, each test batch owns a specific set of initial and target points corresponding to its positioners.\footnote{We conduct the tests on a ASUS ZenBook UX410UAR with an Intel Core i7-8550U @ 1.8GHz x 4 processor, Intel UHD Graphics 620 graphic card on an Microsoft Windows 10, 10.0.15063 version.} The full specifications and the resulting number of the converged positioners and the convergence times are reflected in table \ref{tbl:c12}.
\begin{table}
	\hspace*{-3cm}
	\begin{tabular}{cc  cccccc   cccccc }
		\toprule
		\toprule
		&& \multicolumn{6}{c}{\bfseries Test batch 1} &\multicolumn{6}{c}{\bfseries Test batch 2}\\
		\cline{1-14} 
		{\bfseries Total positioners (\#)} 
		&&52&106&234&449&730&980&54&114&250&481&773&1006\\
		\cline{1-14}
		\multirow{2}{1.7cm}{\centering\bfseries Converged positioners (\#)} & APF  &50&90&196&382&621&844&52&105&228&434&675&889\\
		\cline{2-14}
		& CAPF &52&106&234&449&730&980&54&114&\z{250}&481&473&1006\\
		\cline{1-14}
		\multirow{2}{1.7cm}{\centering\bfseries Convergence time (sec.)} & APF &14.8&36.3&89.7&173.2&317.6&386.9&13.9&31.3&85.1&171.2&267&364.6\\
		\cline{2-14}
		& CAPF &16.7&49.3&96.1&199.8&359.4&503.5&14.7&41.8&99.7&194&303.6&547.9\\
		\cline{1-14}
		$\lambda_{1}$ &&1&1&1&1&1&1&1&1&1&1&1&1\\
		$\lambda_{2}$ &&0.05&0.05&0.05&0.05&0.05&0.05&0.05&0.05&0.05&0.05&0.05&0.05\\
		$\lambda_{3}$ (specific to CAPF) &&0.03&0.03&0.03&0.03&\cellcolor{red!25}0.04&\cellcolor{red!25}0.04&0.03&0.03&0.03&0.03&\cellcolor{red!25}0.04&\cellcolor{red!25}0.04\\
		\bottomrule
		\bottomrule
	\end{tabular}
	\caption{The convergence rate and the convergence time corresponding to test batch 1 and test batch 2 (The highlighted entries are the modified \z{values}, so that the completeness conditions associated with their corresponding test cases are satisfied.)}
	\label{tbl:c12}
\end{table}
The graphical representations of the convergence rates and the convergence times corresponding to test batch 1 and test batch 2 are illustrated in Fig. \ref{fig:8} and \ref{fig:9}, respectively.  

We chose $\lambda_{1} = 1$, $\lambda_{2}= 0.05$, and $\lambda_{3}= 0.03$ for our tests. However, these values do not fulfill the completeness condition corresponding to some of the test cases. So, we used the parameter modification procedure as explained in Sec. \ref{subsec:compSeek}. In particular, the $5$\textsuperscript{th} and the $6$\textsuperscript{th} test cases of both the test batches cannot be completely coordinated by the quoted weighting factor parameters. These cases are highlighted in the last row of table \ref{tbl:c12}. Thus, we modified $\lambda_{3}$ value which ended up with the complete coordinations in those cases. 

The results witness the completeness of the considered test cases using our cooperative navigator (see, Fig. \ref{fig:5}) which indicates the efficiency of our approach. As discussed in Sec. \ref{sec:diss}, the imposed necessity of completeness to the overall coordination process practically gives rise to longer movements and interactions between positioners. So, the trade-off between the improved convergence rate and the longer convergence time leads to the following conclusion: the available time between two consecutive observations may be shorter than the required time for the complete coordination of positioners. In this case, one has to use the competitive navigator (see, Fig. \ref{fig:4}).
\begin{figure}
	\hspace*{-5mm}
	\includegraphics[scale=0.95]{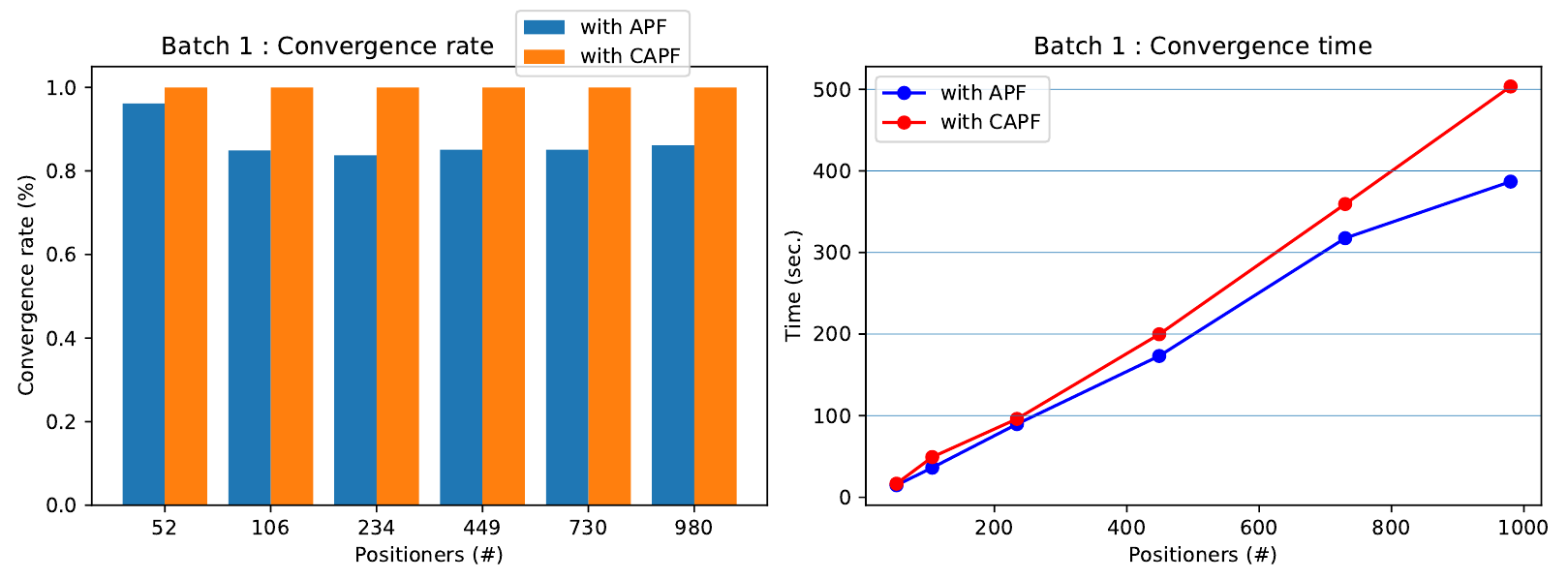}
	\caption{\z{V}isual illustrations of the convergence rate and the convergence time corresponding to test batch 1. (a) The convergence rate corresponding to test batch 1. (b) The convergence time corresponding to test batch 1.}
	\label{fig:8}
\end{figure}
\begin{figure}
	\hspace*{-5mm}
	\includegraphics[scale=0.95]{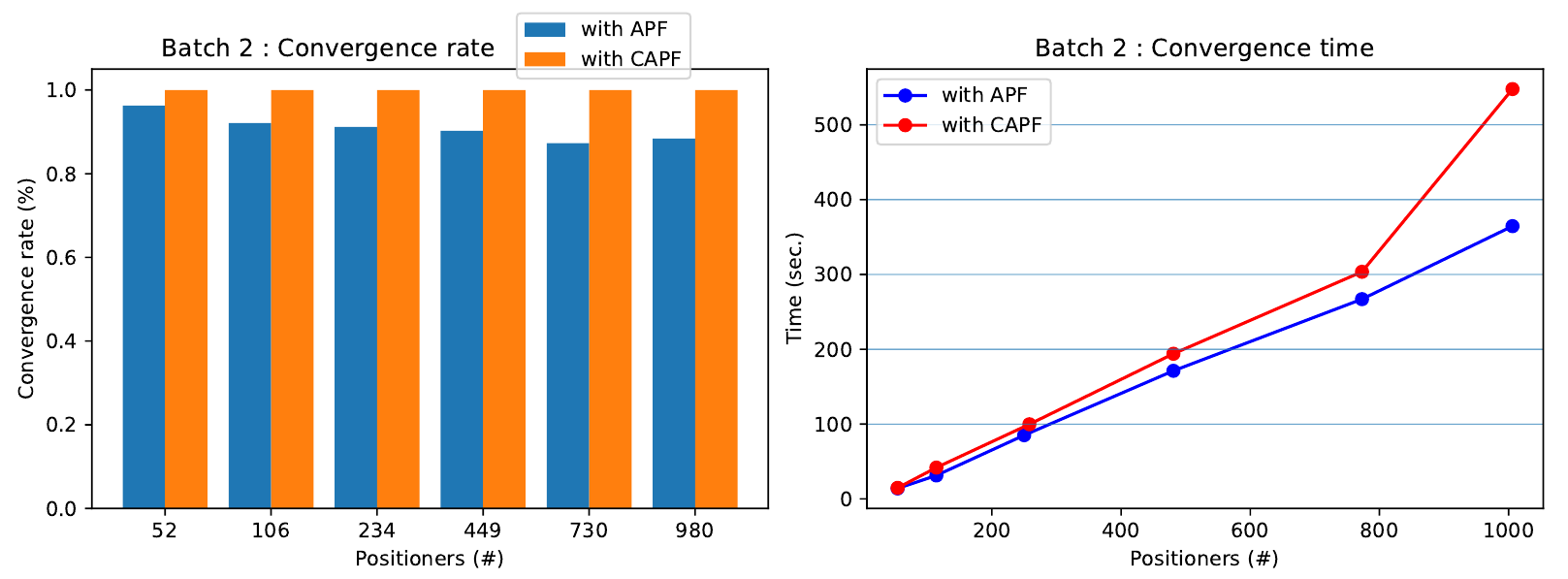}
	\caption{\z{V}isual illustrations of the convergence rate and the convergence time corresponding to test batch 2. (a) The convergence rate corresponding to test batch 2. (b) The convergence time corresponding to test batch 2.}
	\label{fig:9}
\end{figure}
\section{Conclusions}
\label{sec:conc}	
This report studied the completeness problem corresponding to the coordination of robotic optical fiber positioners. In particular, we partitioned the complicated global completeness problem into a set of relatively simpler local completeness problems. We proposed a new artificial potential field by which the completeness of a positioner and its neighboring positioners are cooperatively into account. Then, we found a completeness condition for the local completeness problem, and we showed that the simultaneous fulfillment of all those conditions associated with a positioners set in fact guarantees the global completeness of the overall system. We also presented a completeness-seeking procedure to modify a system's parameters in case the system encounters an incomplete coordination. We obtained the complete coordination at the cost of longer coordination times compared to the required times using a common artificial potential field without cooperation mechanism.
\section{acknowledgments} 
The authors would like to appreciate the anonymous reviewers whose helpful comments gave rise to the improvement of this paper.  
\bibliography{report}   
\bibliographystyle{IEEEtran}  


\vspace{2ex}\noindent\textbf{Matin Macktoobian} received his B.Sc. in computer engineering from KNTU, Tehran, Iran, in 2013. He then obtained his M.A.Sc. in electrical and computer engineering at the University of Toronto, Toronto, Canada, in 2018. He has been currently seeking a Ph.D. in robotics, control, and intelligent system at EPFL, Lausanne, Switzerland, from 2018. His research interests include space robotics, engineering astrodynamics, and astronomical instrumentation. He is the recipient of many scholarships and awards such as the University of Toronto graduate fellowship, the Connaught international scholarship, School of Graduate Studies Conference Grant, and the gold leaf award of the PRIME 2019.

\vspace{2ex}\noindent\textbf{Denis Gillet} received his M.Sc. in electrical engineering and a Ph.D. degree in information systems in 1988 and 1995, respectively, both from EPFL. He was a research fellow at the information systems laboratory of Stanford university, US. He is currently a lead researcher and faculty member at the EPFL school of engineering, where he leads the React multi-disciplinary research group. Dr. Gillet is an executive of the STELLAR European network of excellence in technology enhanced learning. 

\vspace{2ex}\noindent\textbf{Jean-Paul Kneib} holds an M.Sc. degree in astrophysics and space technology and a Ph.D. degree in astrophysics. He has worked as a support astronomer, at ESO in Chile. He has conducted research in gravitational lensing and cosmology in Cambridge, Toulouse, Caltech and Marseille before coming to EPFL. He is currently strongly involved in the Euclid space mission. He has been a member of the ESA astronomy working group, and of the Hubble space telescope user committee. 

\end{document}